\documentclass[reprint,amsmath,amssymb,amsfonts,aps,superscriptaddress,nofootinbib]{revtex4-1}
\usepackage{amsthm}
\usepackage{latexsym}
\usepackage{bbm,dsfont}
\usepackage{color}
\usepackage{graphicx}
\usepackage{epstopdf}
\usepackage{hyperref}
\usepackage{enumerate}

\newtheorem{theorem}{Theorem}

\newtheorem{proposition}{Proposition}
\theoremstyle{definition}
\newtheorem{example}{Example}

\newcommand{\R}{\mathbb R} 
\newcommand{\C}{\mathbb C} 

\newcommand{\hi}{\mathcal{H}} 
\newcommand{\ki}{\mathcal{K}} 
\newcommand{\lh}{\mathcal{L(H)}} 
\newcommand{\lk}{\mathcal{L(K)}} 
\newcommand{\trh}{\mathcal{T(H)}} 
\newcommand{\trhk}{\mathcal{T}(\mathcal{H}\otimes \mathcal{K})} 
\newcommand{\sh}{\mathcal{S(H)}} 
\newcommand{\sk}{\mathcal{S(K)}} 
\newcommand{\ph}{\mathcal{P(H)}} 
\newcommand{\pk}{\mathcal{P(K)}} 

\newcommand{\ip}[2]{\left\langle\,#1\,|\,#2\,\right\rangle} 
\newcommand{\ket}[1]{|#1\rangle} 
\newcommand{\bra}[1]{\langle#1|} 
\newcommand{\tr}[1]{\textrm{tr}\left[#1\right]} 
\newcommand{\id}{\mathbbm{1}} 
\newcommand{\mc}[1]{\mathcal{#1}} 

\newcommand{\Ao}{\mathsf{A}}
\newcommand{\Eo}{\mathsf{E}}
\newcommand{\Po}{\mathsf{P}}
\newcommand{\Qo}{\mathsf{Q}}
\newcommand{\So}{\mathsf{S}}
\newcommand{\Zo}{\mathsf{Z}}

\begin{document}
\title[]{WAY beyond conservation laws$^\star$}
\author{Mikko Tukiainen}
\address{Turku Centre for Quantum Physics, Department of Physics and Astronomy, University of Turku, Finland}
\email{mikko.tukiainen@utu.fi}

\begin{abstract} 
The ability to measure every quantum observable is ensured by a fundamental result in quantum measurement theory. Nevertheless, additive conservation laws associated with physical symmetries, such as the angular momentum conservation, may lead to restrictions on the measurability of the observables. Such limitations are imposed by the theorem of Wigner, Araki and Yanase (WAY). In this paper a new formulation of the WAY-theorem is presented rephrasing the measurability limitations in terms of quantum incompatibility. This broader mathematical basis enables us to both capture and generalise the WAY-theorem by allowing to drop the assumptions of additivity and even conservation of the involved quantities. Moreover, we extend the WAY-theorem to the general level of positive operator valued measures.

$^\star$ Dedicated to my beloved daughter Minttu; may your life be full of joy and wonder.
\end{abstract}
\maketitle

\section{Introduction}

Measurability of physical quantities is an integral part of any scientific theory. Indeed, the whole endeavour of understanding natural phenomena relies crucially on the ability to assign values to the physical properties of the system of interest by means of performing measurements. As physical processes, measurements are subjected to and constrained by the laws of physics. In particular, it is known that quantum theory together with certain conservation laws can set limitations to the measurability of the quantum observables. More specifically, an observable which does not commute with an additive conserved quantity does not admit a repeatable and perfectly precise measurement -- this limit is known as the Wigner-Araki-Yanase (WAY) theorem \cite{Wigner52, Araki60}.
 
In recent investigations the original theorem of WAY has been generalised to more widely applicable contexts. In particular, ways to omit the assumptions of the repeatability of the measurement \cite{Miyadera06, Loveridge10, QTM} and the additivity of the conserved quantity \cite{Kimura08} have been reported. Moreover, different quantitative generalisations of the WAY-theorem which relax the assumption of perfect precision have been studied in Refs.\,\cite{Ozawa02, Ozawa03, Miyadera06, Kimura08, Loveridge10, Busch11}. With all its extensions the current form of the WAY-theorem covers a large class of physically important scenarios and consequently has applications, not solely in quantum measurement theory, but also in the fields of quantum information processing and quantum control. For example, limitations on the realisability of quantum logic gates due to the WAY-theorem have been discussed in Refs.\,\cite{Ozawa022, Ozawa03, Karasawa07, Karasawa09}. 

Even though the WAY-theorem has certainly been extended from the days of its inception, its full scope is still unknown and the formalism of WAY is not particularly intuitive. Our first main result is to introduce a new extension of the WAY-theorem, in which the assumptions of repeatability of the measurement and, not only the additivity, but even the central assumption of conservation of a quantity commuting with the measured observable can be omitted. Formally, our result states that, whenever a quantity commutes with the evolved pointer of the apparatus, a part of it also necessarily commutes with the measured observable. In other words, in our explanation the WAY-theorem can be understood as a consequence of quantum {\it compatibility} \cite{Heinosaari14, Heinosaari16} of a given quantity with the evolved pointer partially inherited by the measured observable. We believe that the intuition behind this formalism is conceptually clearer than in the preceding formulations listed above. We also present examples that demonstrate the limitations of measurability, even if the assumptions of the original WAY-theorem are violated.

Strictly speaking, the restrictions posed by the WAY-theorem affect only the special class of quantum observables associated with the normalised projection valued measures. In vague terms, these {\it sharp} observables correspond to ideally precise measurements and the limitations of WAY-theorem may, in principle, be circumvented by introducing an arbitrarily small amount of inaccuracy in the measured observable: such imprecision can be described by associating the measured observable with a normalised positive operator valued measure. Since noise is inevitably present in every real experiment, from a practical point of view the measured observables are generally {\it unsharp} and it may seem that the WAY-theorem exists only as a theoretical phenomenon. 

It is our second main result to show that quantitative versions of the WAY-theorem persist also in the level of the unsharp observables. In particular, we reveal a natural relation, in which the sharpness of the measured observable and the amount of compatibility of the evolved pointer with a given (additive conserved) quantity govern the WAY-type limitations. These results are then applied to expose restrictions in quantum programming.

\section{WAY limitations}

We begin by outlining the basic concepts of quantum measurement theory relevant to our investigation. Let $\hi$ be a complex separable, possibly infinite dimensional, Hilbert space associated to a quantum system and denote by $\lh$, $\ph$ and $\trh$ the set of bounded operators, projections and trace-class operators on $\hi$, respectively. The identity operator in $\lh$ is denoted by $\id$. The properties of a quantum system are encoded in a quantum state $\varrho$, a positive operator in $\trh$ with $\tr{\varrho}=1$. Quantum states on $\hi$ comprise a convex set that is denoted by $\sh$. The extremal elements of $\sh$ are called pure and any such state is of the form $\ket \varphi \bra \varphi$ for some unit vector $\varphi \in \hi$; for this reason, we can call any unit vector a quantum state without risk of confusion.

Let $\Omega$ be a set and $\Sigma$ a $\sigma$-algebra of subsets of $\Omega$. We associate {\it quantum observables} with normalised positive operator valued measures (POVMs) $\Eo:\Sigma \to \lh, X \mapsto \Eo(X)$. The number $p^\Eo_\varrho(X) = \tr{\Eo(X) \, \varrho}$ is interpreted as the probability that a measurement of $\Eo$ performed on $\varrho\in\sh$ leads to a result in $X\in \Sigma$. We call the operators $\Eo(X)$ in the range of an observable {\it effects}. An observable whose all effects are multiples of the identity, that is $\Eo(X) = p(X)\, \id$ for some probability measure $p:\Sigma \to [0,1]$, is called {\it trivial}. The normalised projection valued measures (PVMs) $\Ao:\Sigma \to \ph$ are called {\it sharp observables}. If $\Omega=\{x_1,\,x_2,...\}$ with $n(\leq\infty)$ elements and $\Sigma= 2^{\Omega}$ is the corresponding outcome space of an observable $\Eo$, we say that $\Eo$ is a {\it discrete} ($n$-valued) observable. In particular, for any $\vec{m} = (m_x,m_y,m_z) \in \R^3$, $|| \vec{m}||\leq 1$, we define the discrete 2-valued spin-observables $\So_{\vec{m}}:2^{\{+,-\}} \to \mc L(\C^2)$ via $\So_{\vec{m}}(\pm) = \frac{1}{2}\big( \id \pm \vec{m} \cdot \vec{\sigma})$, where $\vec{m} \cdot \vec{\sigma} = \sum_{i=x,y,z} m_i \sigma_i$. Here, $\sigma_x, \sigma_y$ and $\sigma_z$ are the Pauli spin-matrices
\begin{eqnarray}
\sigma_x = \left(\begin{array}{cc}
0 & 1 \\
1 & 0
\end{array}\right), \ 
\sigma_y = \left(\begin{array}{cc}
0 & -i \\
i & 0
\end{array}\right), \
\sigma_z = \left(\begin{array}{cc}
1 & 0 \\
0 & -1
\end{array}\right). \quad
\end{eqnarray} 
An observable $\So_{\vec{m}}$ is sharp exactly when $|| \vec{m}||=1$. We will use the notations $\hat x=(1,0,0),$ $\hat y=(0,1,0)$ and $\hat z=(0,0,1)$. 

The quantum description of a {\it measurement} is mathematically encoded in a 4-tuple $\langle \ki, \Zo, \mc V, \xi \rangle$, where $\ki$ is the Hilbert space associated to the measurement apparatus, $\Zo:\Sigma \to \mc L(\ki)$ is the pointer observable, the completely positive trace preserving (CPTP) linear map $\mc V: \trhk\to\trhk$ describes the measurement coupling and $\xi \in \mc S(\ki)$ is the initial state of the apparatus. Under the measurement process the initially separable compound state of system and apparatus $\varrho \otimes \xi$ evolves to $\mc V(\varrho\otimes \xi)$ and after the evolution the measurement outcome is read from the pointer scale. The observable $\Eo:\Sigma\to\lh$ measured in $\langle \ki, \Zo, \mc V, \xi \rangle$ is reproduced by the formula $p^\Eo_\varrho(X) = \tr{\id \otimes \Zo(X)\, \mc V(\varrho \otimes \xi)}$ required to hold for all $X\in\Sigma$ and $\varrho \in \sh$  \cite{QTM}. 

In this study we will focus on measurements of a particular form, viz. assuming that $\Zo: \Sigma \to \pk$ is sharp, $\mc V$ is conjugation with a unitary operator $U$ on $\hi \otimes \ki$ and $\xi = \ket \phi \bra \phi$ for some unit vector $\phi\in\ki$: such {\it normal} measurements \cite{QTM} we write as $\langle \ki, \Zo, U, \phi \rangle$. From the physical point of view, normal measurements describe ideally functioning measuring devices, where the system-apparatus-composite forms a closed quantum system, the detectors work with perfect accuracy and the apparatus is initially prepared in a state of maximal information. The observable $\Eo$ measured in a normal measurement attains a simple form
\begin{eqnarray}\label{eq:mntobs}
\Eo(X) = V_\phi^* U^*\big( \id \otimes \Zo(X)\big) U V_\phi, \qquad X \in \Sigma,
\end{eqnarray}
where the isometry $V_\phi: \hi \to \hi\otimes \ki$ is defined via $V_\phi(\varphi) = \varphi\otimes \phi$ for all $\varphi \in \hi$. The measured observable is sharp exactly when $[U^*\big( \id \otimes \Zo(X)\big) U, V_\phi V_\phi^*]=0$ for all $X \in \Sigma$ \cite{Lahti04}; this result will be used frequently during the present work. 

A normal measurement $\langle \ki, \Zo, U, \phi \rangle$ is {\it repeatable} if any recorded outcome of the measurement does not change upon its immediate repetition, or equivalently  $\Eo(X) =  V_\phi^* U^* \big( \Eo(X) \otimes \Zo(X) \big) U V_\phi$, for all $X \in \Sigma$. Not all measurements are repeatable and furthermore not all observables even admit repeatable measurements. Indeed, only discrete observables whose all non-zero effects have eigenvalue $1$ can be realized in a repeatable normal measurement \cite{Ozawa84}.

A fundamental result in quantum measurement theory ensures that every quantum observable can be realised in a measurement, even in a normal one \cite{ Ozawa84}. As physical processes, however, measurements are subjected to laws of quantum physics that can impose restrictions on the measurability of observables. One such limitation, first pointed out in measurements of spin-$\frac 12$-systems by Wigner \cite{Wigner52} and later stated in more general setting by Araki and Yanase \cite{Araki60}, is due to conservation laws for additive quantities that do not commute with the observable to be measured. To make this more exact, we call a bounded selfadjoint operator $L \in \mc L(\hi \otimes \ki)$ a {\it conserved quantity} (with respect to the measurement coupling $U$), if $\tr{L \, \varrho} = \tr{L \, U\varrho U^*}$ for all $\varrho \in \mc S(\hi\otimes \ki)$ or equivalently $[L, U]=0$. If furthermore $L=L_1\otimes \id + \id \otimes L_2$, where $L_1$ and $L_2$ are selfadjoint operators in $\lh$ and $\lk$, respectively, we say that $L$ is an {\it additive conserved quantity}. The theorem of Wigner, Araki and Yanase (WAY) then states that, if $L=L_1\otimes \id + \id \otimes L_2$ is an additive conserved quantity w.r.t. the coupling $U$ of a repeatable normal measurement $\langle \ki, \Zo, U, \phi \rangle$ of a sharp (discrete) observable $\Ao$, then necessarily $[\Ao(X), L_1] =0$ for all $X\in\Sigma$. From this point onwards we will use the shortened notation $[\Eo, L] = 0$ whenever $[\Eo(X),L] =0$ for all $X \in \Sigma$.

Many realistic measurements are not repeatable. However, to ensure a stable record of the measurement, it is often assumed that the pointer reading is subjected to a repeatable measurement. When such is the case, the above WAY-theorem persists at the pointer level implying $[\Zo, L_2]=0$; adapting the terminology used in the literature \cite{Ozawa02, Loveridge10} we will call this commutation the {\it Yanase condition}. Importantly, it has been shown in \cite{QTM, Loveridge10} that the same conclusion of the WAY-theorem can be drawn if the assumption of repeatability is replaced by the Yanase condition. As a summary of the above, we present the following theorem; see Ref.\,\cite{Loveridge10} for the proof.
\begin{theorem}[WAY-theorem]
Let $\langle \ki, \Zo, U, \phi \rangle$ be a normal measurement of a sharp observable $\Ao$ and let $L_1\in\lh$ and $L_2\in\lk$ be bounded selfadjoint operators such that $L=L_1\otimes \id + \id \otimes L_2\in\mc L(\hi\otimes\ki)$ is an additive conserved quantity. Assume that
$\langle \ki, \Zo, U, \phi \rangle$ is repeatable or satisfies the Yanase condition. Then
$[\Ao,L_1] = 0$.
\end{theorem}

An easy check confirms that under the conservation $[L, U]=0$ of an additive quantity $L=L_1\otimes \id + \id \otimes L_2$ the Yanase condition $[\Zo, L_2]=0$ is equivalent to $[U^*\big( \id \otimes \Zo\big) U, L]=0$. In the following, we shall see that the {\it weak Yanase condition} $[U^*\big( \id \otimes \Zo\big) U, L]=0$ may be used to generalise the WAY-theorem. 

\begin{proposition}\label{prop:way} Let $\langle \ki, \Zo, U, \phi \rangle$ be a normal measurement of a sharp observable $\Ao$ and let $L\in\mc L(\hi\otimes\ki)$. If $[U^*\big( \id \otimes \Zo \big) U, L]=0$, then $[\Ao, V_\phi^* L V_\phi]=0$. In particular, if $L_1\in\lh$ and $L_2\in\lk$ are bounded selfadjoint operators such that $L=L_1\otimes \id + \id \otimes L_2$ is an additive conserved quantity and $[\Zo, L_2]=0$, then $[\Ao, L_1] =0$.
\end{proposition}
\begin{proof}
We first recall that $V_\phi$ is an isometric operator, that is $V_\phi^*V_\phi= \id$. Therefore, since $\Ao$ is assumed to be sharp, we have 
\begin{eqnarray}
\Ao(X) \, V_\phi^* L V_\phi &=&  V_\phi^* U^*\big( \id \otimes \Zo(X)\big) U V_\phi \, V_\phi^* L V_\phi \nonumber \\
&=& V_\phi^* U^*\big( \id \otimes \Zo(X)\big) U L V_\phi \nonumber \\
&=& V_\phi^* L V_\phi \,  V_\phi^* U^*\big( \id \otimes \Zo(X)\big) U V_\phi \nonumber \\
&=&  V_\phi^* L V_\phi \, \Ao(X).
\end{eqnarray}
For the second claim, we notice that $ V_\phi^* L V_\phi = L_1 + \ip{\phi}{L_2 \phi} \id$ whenever $L=L_1\otimes \id + \id \otimes L_2$, and that the assumption of $L_2$ being bounded ensures $\ip{\phi}{L_2 \phi}<\infty$.
\end{proof}

The above result establishes the generalised WAY-type limitations that hold for (continuous) observables without the necessity of $L$ being additive, or even conserved. For instance, let us consider a {\it multiplicative} selfadjoint quantity $L=L_1 \otimes L_2\in\mc L(\hi\otimes\ki)$ for which $\ip{\phi}{L_2 \phi} \neq 0$, e.g. $L_2$ is invertible. Supposing the assumptions of Prop.\,\ref{prop:way} hold, the condition $[U^*\big( \id \otimes \Zo \big) U, L]=0$ then implies $[\Ao, L_1] = 0$; a similar result was previously found for multiplicative conserved quantities in \cite{Kimura08}. 

\begin{figure}[t!]
\includegraphics[width=0.5\textwidth]{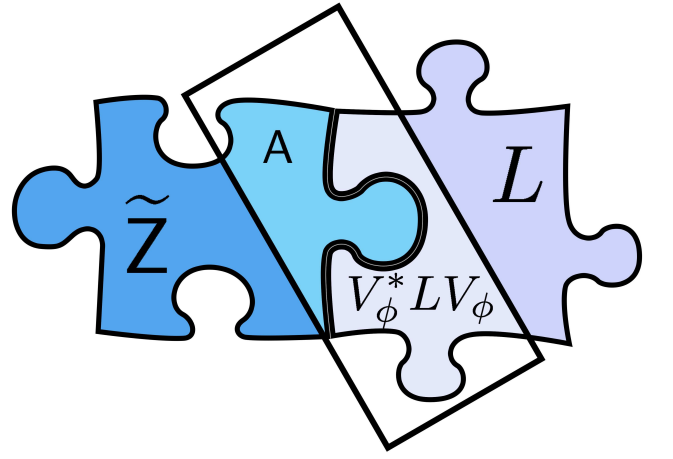}
\caption{Our formalism (Prop.\,\ref{prop:way}) gives the WAY-theorem the following interpretation: in a normal measurement $\langle \ki, \Zo, U, \phi \rangle$ of a sharp observable $\Ao$, the compatibility of the evolved pointer of the measurement apparatus $\widetilde \Zo = U^*\big( \id \otimes \Zo \big) U$ with a quantity $L\in\mc L(\hi\otimes \ki)$ implies the compatibility of $\Ao$ with $V_\phi^* L V_\phi$. To this end, the quantity $L$ does not need to be additive nor conserved.}\label{fig:waycompatibility}
\end{figure}

In an informal manner of speaking, two quantum devices are said to be {\it compatible} if there exists a measurement that is capable of realizing both the devices as its parts; for precise definitions of the terms we refer the reader to Refs.\,\cite{Heinosaari14, Heinosaari16}. For example, two observables are compatible if and only if they are jointly measurable. Accordingly, the compatibility of the (Heisenberg-) evolved pointer $U^*\big( \id \otimes \Zo \big) U$ with a selfadjoint quantity $L\in\mc L(\hi\otimes\ki)$ is equivalent to their commutativity $[U^*\big( \id \otimes \Zo \big) U, L]=0$. This implies that Prop.\,\ref{prop:way}, and therefore the WAY-theorem, may be understood as a consequence of compatibility of the evolved pointer with $L$ inherited by the measured observable; see Fig.\,\ref{fig:waycompatibility}. 

\begin{example}\label{ex:compchannel} Any operator of the form $L= U^* \big( B \otimes \id \big) U$, $B \in \lh$, commutes with $U^*\big( \id \otimes \Zo\big) U$. Assuming that the measured observable $\Ao$ is sharp, Prop.\,\ref{prop:way} implies $[ \Ao, V_\phi^* U^*\big( B \otimes \id \big) U V_\phi]=0$. We note that $\mc E^*(B) := V_\phi^* U^*\big( B \otimes \id \big) U V_\phi$ defines the Heisenberg channel, a completely positive unital linear map $\lh \to \lh$, associated to the measurement $\langle \ki, \Zo, U, \phi \rangle$. The commutativity $[ \Ao(X), \mc E^*(B)]=0$, $X\in\Sigma,$ $B\in\lh$, implied by Prop.\,\ref{prop:way} is then a restatement of the known result that compatibility of a channel with a sharp observable is equivalent to their commutativity \cite{Ozawa84, Pellonpaa13, Haapasalo14, Heinosaari14}.
\end{example}

We fix $\hi=\C^2=\ki$ for the rest of this section and let the following examples further demonstrate the power of Prop.\,\ref{prop:way}.

\begin{example}\label{ex:cunit} Consider a controlled unitary $U=  \id \otimes \ket 0 \bra 0  +  \sigma_z \otimes \ket 1 \bra 1$, where $\ket 0$ and $\ket 1$ are the eigenvectors of the Pauli spin-operator $\sigma_z$. Let $U$ serve as a measurement coupling in $\langle \C^2, \Zo, U, \phi \rangle$, where $\Zo:2^{\{+,-\}}\rightarrow\mc L(\C^2)$ is sharp and $\phi\in\C^2$ is a unit vector. We find that each $L=\text{diag}(a,b) \otimes \id$, $a,b \in \R$ is an additive conserved quantity with $L_2=0$ trivially commuting with any $\Zo$. The WAY-theorem then implies that a sharp observable $\Ao$ realized in $\langle \C^2, \Zo, U, \phi \rangle$ has to satisfy $\big[\Ao,\text{diag}(a,b) \big]=0$ for all $a,b\in \R$. In particular $[\Ao, \So_{\hat z}]=0$, or equivalently $\Ao$ and $\So_{\hat z}$ are jointly measurable. Indeed, it can be confirmed that $\langle \C^2, \Zo, U, \phi \rangle$ realises non-trivial sharp observables only when choosing $\Zo=\So_{\hat x}$ and $\phi = \frac{1}{\sqrt 2}( \ket 0 \pm \ket 1)$ (up to a global phase): the measured sharp observables are $\So_{\pm \hat z}$, respectively. With these choices the corresponding measurements are also repeatable.
\end{example}

\begin{example}\label{ex:genway} Let us fix the unitary 
\begin{eqnarray}
U=\frac{1}{\sqrt 2}\left(\begin{array}{cccc}
i & 0 & 0 & 1 \\
i & 0 & 0 &-1 \\
0 & i & 1 & 0 \\
0 & i &-1 & 0
\end{array}\right)
\end{eqnarray} and consider the measurement $\langle \C^2, \Zo, U, \phi \rangle$ where $\Zo:2^{\{+,-\}}\rightarrow\mc L(\C^2)$ is sharp and $\phi\in\C^2$ is a unit vector. We notice that the measured observable is sharp, that is $[U^*\big( \id \otimes \Zo \big) U, V_\phi V_\phi^*]=0$, for a non-trivial $\Zo$ if and only if $\Zo =  \So_{\pm \hat x}$, regardless of the choice of the unit vector $\phi \in \C^2$. With these choices the measurements are also repeatable. It may, however, be confirmed that the only additive conserved quantities w.r.t. $U$ are of the form $k \id \otimes \id$, $k\in\R$. Therefore, the measurement $\langle \C^2, \Zo, U, \phi \rangle$ is not subjected to non-trivial limitations in the traditional sense of the WAY-theorem for any choices of $\Zo$ and $\phi$. Although $[\Zo, \sigma_z]\neq 0$ in the cases $\Zo =  \So_{\pm \hat x}$, it can be confirmed that the additive quantity $L=\sigma_z \otimes \id + \id \otimes \sigma_z$ commutes with $U^*\big( \id \otimes \Zo\big) U$. Therefore, Prop.\,\ref{prop:way} implies that any sharp observable $\Ao$ realised in  $\langle \C^2, \Zo, U, \phi \rangle$ must satisfy $[\Ao, \sigma_z]=0$, which again is equivalent to $\Ao$ and $\So_{\hat z}$ being jointly measurable. 
\end{example}

\begin{example}\label{ex:swap}
Define a unitary $U=S \,\big(\id \otimes \ket 0 \bra 0 + \frac{1}{\sqrt 2} \left(\begin{array}{cc}
1 & 1 \\
1 & -1
\end{array}\right) \otimes \ket 1 \bra 1 \big),$ where $S$ is the SWAP gate 
\begin{eqnarray}\label{eq:swap}
S=\left(\begin{array}{cccc}
1 & 0 & 0 & 0 \\
0 & 0 & 1 & 0 \\
0 & 1 & 0 & 0 \\
0 & 0 & 0 & 1
\end{array}\right),
\end{eqnarray} and fix $\Zo=\So_{\hat z}$. Again, the only additive conserved quantities w.r.t. $U$ are of the form $k \id \otimes \id$, $k\in\R$ and furthermore any additive quantity $L$ satisfying $\big[U^*\big( \id \otimes \Zo\big) U, L\big]=0$ is trivial on the system side, $\id \otimes \text{diag}(a,b)$, $a,b\in\R$. We conclude that the additive quantities will not set any limitations for the measured observables via Prop.\,\ref{prop:way} in this case. However, it may be confirmed that there exist the two classes of multiplicative selfadjoint quantities $\left(\begin{array}{cc}
a & 0 \\
0 & b
\end{array}\right) \otimes \ket 0 \bra 0$ and $\left(\begin{array}{cc}
a & b \\
b & a
\end{array}\right) \otimes \ket 1 \bra 1$, $a,b\in\R$, commuting with $U^*\big( \id \otimes \Zo\big) U$.
With the above choices, the probe states $\ket 0$ and $\ket 1$ realize sharp observables $\So_{\hat z}$ and $\So_{\hat x}$, respectively, which clearly satisfy the corresponding commutation relations set by Prop.\,\ref{prop:way}.
\end{example}

One can also find cases of normal measurements which are not subjected to the WAY-type limitations even in the more general sense of Prop.\,\ref{prop:way}. Consider again the controlled unitary of Ex.\,\ref{ex:cunit}, $U_1= \id \otimes \ket 0 \bra 0  +  \sigma_z \otimes \ket 1 \bra 1$. It can be concluded that $U_1$ may be used to realize non-trivial sharp observables on $\hi$ by choosing the pointer $\Zo$ and the probe state $\phi$ appropriately, for instance $\Zo=\So_{\hat x}$ and $\phi=\frac{1}{\sqrt 2}\big(\ket 0 + \ket 1 \big)$. By interchanging the roles of the system and the apparatus, $\langle \hi, \Zo, U_1, \phi\rangle$ realizes a sharp observable $\Ao:\Sigma\rightarrow\mc P(\ki)$ also on $\ki$ defined via $\tr{\Ao (X) \, \xi} = \tr{\Zo(X) \otimes \id \, U_1\big( \ket \phi \bra \phi \otimes \xi \big) U_1^*}$, for all $\xi\in\sk$: this is immediately verified by noticing that $U_1$ commutes with the SWAP gate $S$ defined in Eq.\,\eqref{eq:swap}. Defining $U_2= \id \otimes \ket 0 \bra 0 +  \left(\begin{array}{cc}
0 & i \\
1 & 0
\end{array}\right) \otimes \ket 1 \bra 1$, however, does not serve as an interaction in a normal measurement $\langle \ki, \Zo, U_2, \phi \rangle$ of any non-trivial sharp observable on $\hi$ with any choices of $\Zo$ and $\phi$. On the other hand, $\langle \hi, \Zo, U_2, \phi \rangle$ may be used to realize non-trivial sharp observables on $\ki$, for example by choosing $\Zo=\So_{\hat z}$ and $\phi = \ket 0$. Finally, the unitary $U_3= \id \otimes \ket 0 \bra 0 +  \left(\begin{array}{cc}
1 & 0 \\
0 & i
\end{array}\right) \otimes \ket 1 \bra 1$ cannot be used to realize any non-trivial sharp observables on either $\hi$ or $\ki$. We summarise, that a coupling may be subjected to WAY-type limitations imposed by Prop.\,\ref{prop:way} both two-sidedly (e.g. $U_1$), only one-sidedly ($U_2$) or not be subjected to such limitations at all ($U_3$), simply due to its ability to serve as a measurement coupling for non-trivial sharp observables.

\section{Generalisation to POVMs}

As already discussed above, a way to circumvent the limitations set by WAY is to consider, instead of sharp, general (smeared) POVMs as measured observables. This deviation from the PVM-picture is often even reasonable from the physical point of view, as imperfections are present in all realistic measurement implementations. Therefore, since restricting one's attention only to sharp observables would make the WAY-theorem an unphysical curiosity, any step of developing the WAY-theorem in the broader context of {\it unsharp} observables is well justified from both theoretical and practical standpoints.

We will next elucidate that the limitations posed by Prop.\,\ref{prop:way} persist also at the level of POVMs. 
\begin{widetext}
\begin{proposition}\label{prop:wayofpovms1}
Let $\langle \ki, \Zo, U, \phi \rangle$ be a normal measurement of an observable $\Eo:\Sigma \to \lh$. Then, for all selfadjoint $L \in \mc L (\hi \otimes \ki)$, the inequality
\begin{eqnarray}\label{ineq:wayofpovms1}
\left|\left| \big[ \Eo(X), V_\phi^* L V_\phi \big] \right|\right| &\leq & 2 \left|\left| \big[ U^* \big( \id \otimes \Zo(X) \big) U, V_\phi V_\phi^* \big] \right|\right|  \, ||L|| + \left|\left| \big[ U^* \big( \id \otimes \Zo(X) \big) U, L \big] \right|\right| 
\end{eqnarray}
holds for all $X \in \Sigma$. 
\end{proposition}
\begin{proof}
We first notice that $i \big[ A, B \big]$ is a bounded selfadjoint operator for all bounded selfadjoint $A,B\in\lh$. Furthermore, $ V_\phi^* L V_\phi$ is a bounded selfadjoint operator on $\hi$ whenever $L$ is bounded and selfadjoint on $\hi\otimes \ki$. Therefore

\begin{eqnarray}
\left|\left| \big[ \Eo(X),  V_\phi^* L V_\phi \big] \right|\right|  &=& \sup_{|| \varphi ||\leq 1} \left| \ip{\varphi}{ \big[ V_\phi^* U^* \big( \id \otimes \Zo(X) \big) U V_\phi,  V_\phi^* L V_\phi \big] \varphi}\right| \nonumber \\
&=& \sup_{|| \varphi ||\leq 1}  \left| \ip{\varphi}{ V_\phi^* \left( U^* \big( \id \otimes \Zo(X) \big) U V_\phi V_\phi^* L - L V_\phi V_\phi^* U^* \big( \id \otimes \Zo(X) \big) U \right) V_\phi \varphi} \right| \nonumber \\
&=& \sup_{|| \varphi ||\leq 1}  \left| \ip{\varphi\otimes\phi}{ \big[U^* \big( \id \otimes \Zo(X) \big) U, V_\phi V_\phi^*\big] L \, \varphi\otimes\phi} + \ip{\varphi\otimes\phi}{ \big[U^* \big( \id \otimes \Zo(X) \big) U, L  V_\phi V_\phi^*\big] \varphi\otimes\phi} \right| \nonumber \\
&\leq & 2 \left|\left| \big[ U^* \big( \id \otimes \Zo(X) \big) U, V_\phi V_\phi^* \big] \right|\right|\, ||L|| + \left|\left| \big[ U^* \big( \id \otimes \Zo(X) \big) U, L \big] \right|\right|,
\end{eqnarray}
where we have used the fact that $V_\phi^* V_\phi V_\phi^* U^* \big( \id \otimes \Zo(X) \big) U L V_\phi = V_\phi^* U^* \big( \id \otimes \Zo(X) \big) U L V_\phi V_\phi^* V_\phi$, elementary commutation relations, the triangle inequality and the Cauchy-Schwarz inequality.
\end{proof}
\end{widetext}

In the right-hand-side of the Ineq.\,\eqref{ineq:wayofpovms1} one recognises the two terms: the first one related to the ``sharpness'' of the measured observable \cite{Lahti04} and the second one to the weak Yanase condition. Proposition \ref{prop:way} follows as a corollary exactly when these two terms vanish. 

There are also different WAY-type limitations to be found, as will be shown in the following. The proof is similar to that of the previous proposition and we will omit it.

\begin{proposition}\label{prop:wayofpovms2}
Let $\langle \ki, \Zo, U, \phi \rangle$ be an $\Eo$-measurement. Then, for all selfadjoint $L \in \lh$, the inequalities
\begin{eqnarray}\label{ineq:wayofpovms2}
\left|\left| \big[ \Eo(X), L\big] \right|\right| &\leq & \left|\left| \big[ U^* \big( \id \otimes \Zo(X) \big) U, L \otimes \id \big] \right|\right| \nonumber \\
&\leq & 2 \left|\left| [ U, L \otimes \id ] \right|\right|
\end{eqnarray}
hold for all $X \in \Sigma$.
\end{proposition}

Propositions \ref{prop:wayofpovms1} and \ref{prop:wayofpovms2} become particularly powerful in the cases where their right-hand-sides vanish. Although the two results have apparent similarity, the limitations set by them can be very different. We present the following examples for clarification.

\begin{figure*}[t!]
\begin{minipage}[h]{0.95\textwidth}
\includegraphics[width=0.95\textwidth]{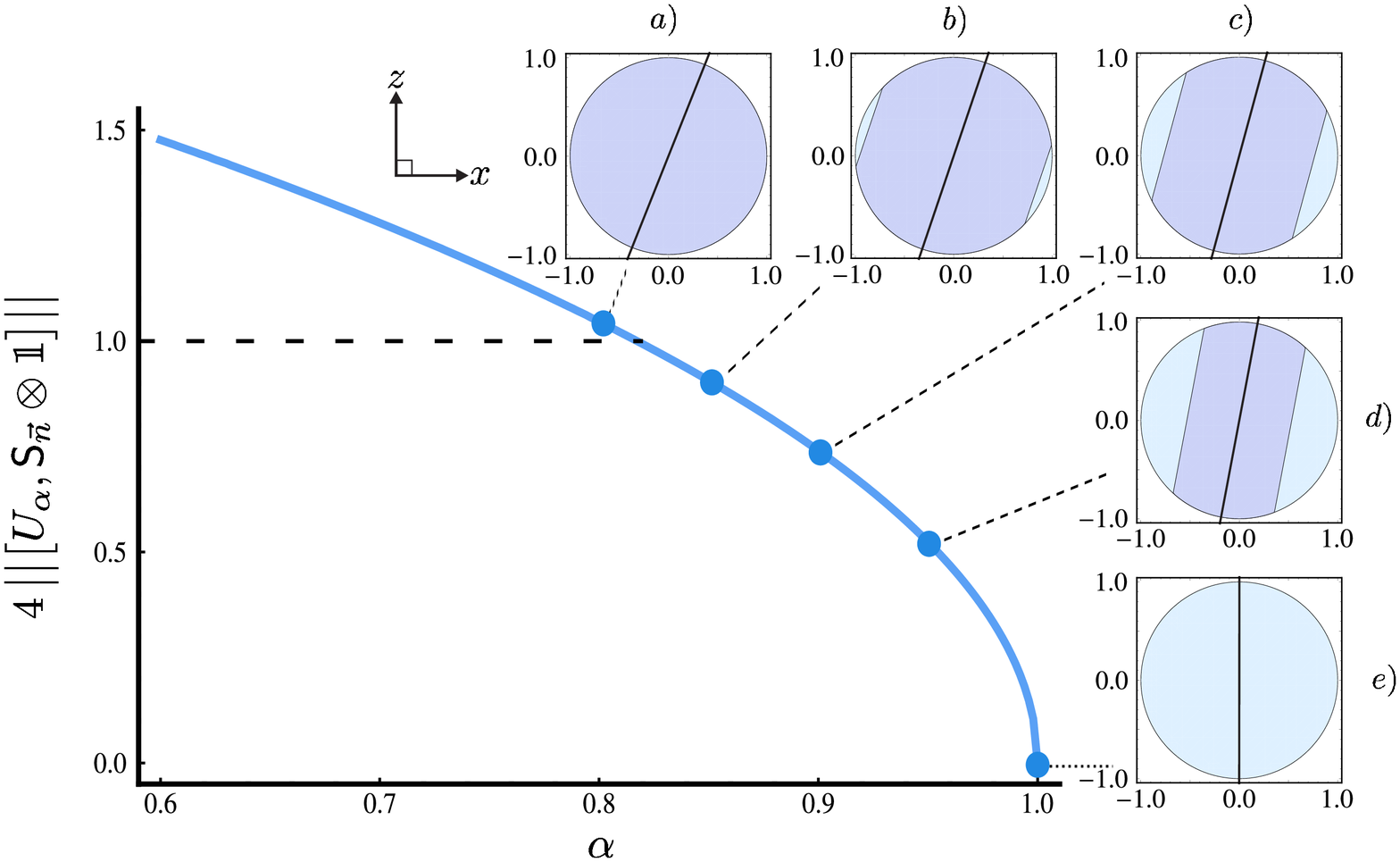}
\caption{The value of the quantity $4 \left|\left| \big[U_\alpha, \So_{\vec n} \otimes \id \big] \right|\right|$ minimised over $\vec n$, where $U_\alpha$ is as defined in Ex.\,\ref{ex:povms4} and $\So_{\vec n}(\pm) = \frac{1}{2}\big( \id \pm \vec n \cdot \vec \sigma \big)$, $||\vec n||=1,$ is plotted in terms parameter $\alpha$ ranging from $0.6$ to $1$. Limitations for the measurability of the observables $\So_{\vec m}(\pm)=\frac{1}{2}\big( \id \pm \vec m \cdot \vec \sigma \big)$, $||\vec m||\leq1,$ set by the relation $2||\big[\So_{\vec m},\So_{\vec n} \big] ||= || \vec m \times \vec n || \leq 4 \left|\left| \big[U_\alpha, \So_{\vec n} \otimes \id \big]  \right|\right|$, are present whenever $4 \left|\left| \big[U_\alpha, \So_{\vec n} \otimes \id \big]  \right|\right|<1$. These limitations have been illustrated by mapping the cross-sections in $xz$-plane of the effects that are at least in principle realisable with $U_\alpha$ for five different values of $\alpha$: a) $\alpha = 0.8$, b) $\alpha = 0.85$,  c) $\alpha = 0.9$,  d) $\alpha = 0.95$ and e) $\alpha = 1$. The total set of effects can attained by rotating the cross-sections a)-e) about the corresponding symmetry axes that have been depicted as black lines.   }\label{fig:wayforpovms}
\end{minipage}
\end{figure*}

\begin{example} Let us denote by $\bar \Ao$ the selfadjoint operator defined as the first moment of the PVM $\Ao:\mc B(\R) \to \mc P (L^2(\R))$: $\bar \Ao = \int_\R x \, \Ao(dx)$. Assume that one intends to measure a sharp observable $\Ao$ by coupling it to the momentum observable $\Po$ of the apparatus via the unitary interaction $U=\text{exp}(i \lambda \bar \Ao \otimes \bar \Po)$, where the parameter $\lambda\in\R$ quantifies the strength of the measurement interaction. One natural choice for the pointer in order to monitor the shifts generated by $U$ is the position observable $\Qo$ of the apparatus. The resulting {\it standard model} of measurement $\langle L^2(\R), \Qo, U, \phi \rangle$ is one of the most widely used forms of normal quantum measurements \cite{Busch95, QTM}. The actual observable measured in this process is $\Eo(X) = \int_\R p^\Qo_\phi (X-\lambda x)\, \Ao(d x)$. As such, $\Eo$ is a {\it smeared} unsharp version of the intended sharp observable $\Ao$. 

The observables $\Eo$ and $\Ao$ are clearly jointly measurable: $[\Eo(X), \Ao(Y)] = 0,$ for all $X,Y \in \mc B(\R)$. In fact since $\big[U, \Ao \otimes \id\big]=0$, Prop.\,\ref{prop:wayofpovms2} implies that all the observables realisable with this coupling are jointly measurable with $\Ao$, regardless of the choice of the pointer observable and the probe state. This same conclusion cannot be generally drawn from Prop.\,\ref{prop:wayofpovms1}. Namely, since the measured observable $\Eo$ in a standard model can be sharp only if $\Ao$ is discrete \cite{Busch95}, the ``sharpness'' term in Ineq.\,\eqref{ineq:wayofpovms1} is generally non-vanishing.
\end{example}

For the rest of the examples of this section we will again fix $\hi=\C^2=\ki$.

\begin{example}\label{ex:povm1}
Recall the coupling $U_3= \id \otimes \ket 0 \bra 0 +  \left(\begin{array}{cc}
1 & 0 \\
0 & i
\end{array}\right) \otimes \ket 1 \bra 1$ introduced above as a measurement coupling between one-qubit system and one-qubit apparatus. It may be confirmed that $[U, L \otimes \id]=0$ for any $L=\left(\begin{array}{cc}
a & 0 \\
0 & b
\end{array}\right),$ $a,b\in\R$, and Prop.\,\ref{prop:wayofpovms2} implies that $[\Eo, \So_{\hat z}]=0$, that is all the measured observables realisable with this coupling are always jointly measurable with $\So_{\hat z}$. However, as mentioned before, it is not possible to use $U_3$ as a coupling in a measurement of any non-trivial sharp observable. Therefore, the right-hand-side of Ineq.\,\eqref{ineq:wayofpovms1} is always non-vanishing and Prop.\,\ref{prop:wayofpovms1} fails to reproduce the same conclusion.
\end{example}

\begin{example}
Consider again the coupling $U=S \,\big(\id \otimes \ket 0 \bra 0 + \frac{1}{\sqrt 2} \left(\begin{array}{cc}
1 & 1 \\
1 & -1
\end{array}\right) \otimes \ket 1 \bra 1 \big)$ and fix $\Zo=\So_{\hat z}$. It was pointed out in Ex.\,\ref{ex:swap} that, for a selfadjoint $L\in\mc L (\hi)$, both the relations $\big[U, L\otimes \id\big] =0$ and $\big[ U^* \big( \id \otimes \Zo \big) U, L \otimes \id \big]=0$ imply that $L$ is trivial. Accordingly, either of the right-hand-sides in Ineq.\,\eqref{ineq:wayofpovms2} vanishes if and only if $L$ is trivial. On the other hand, the limitations set by the two classes of multiplicative quantities, $\left(\begin{array}{cc} 
a & 0 \\
0 & b
\end{array}\right) \otimes \ket 0 \bra 0$ and $\left(\begin{array}{cc}
a & b \\
b & a
\end{array}\right) \otimes \ket 1 \bra 1$, $a,b\in\R$, that were pointed out in Ex.\,\ref{ex:swap} are captured by Prop.\,\ref{prop:wayofpovms1}: the right-hand-side of Ineq.\,\eqref{ineq:wayofpovms1} vanishes for both of these classes.
\end{example}

\begin{example}\label{ex:povms4} Finally, let us consider the measurability of  qubit observables $\So_{\vec{m}}(\pm) = \frac{1}{2}\big( \id \pm \vec{m}\cdot \vec{\sigma} \big)$, $\vec m \in \R^3$,  $||\vec{m}||\leq 1$, with a unitary coupling $U_\alpha=
\left(\begin{array}{cc}
\alpha & \sqrt{1-\alpha^2} \\
\sqrt{1-\alpha^2} & -\alpha
\end{array}\right) \otimes \ket 0 \bra 0 + \left(\begin{array}{cc}
1 & 0 \\
0 & i
\end{array}\right) \otimes \ket 1 \bra 1 $, where $0 \leq \alpha \leq 1$. For any sharp $\So_{\vec{n}} (\pm) = \frac{1}{2}\big( \id \pm \vec{n}\cdot \vec{\sigma} \big),$ $\vec n \in \R^3$, $||\vec{n}||=1$, the quantity $2 \left|\left| \big[ \So_{\vec m}, \So_{\vec n}] \right|\right| = || \vec m \times \vec n ||$ constitutes a measure of incompability of the qubit observables $\So_{\vec m}$ and $\So_{\vec n}$ vanishing for compatible observables and attaining its maximum value $1$ for the maximally incompatible ones, i.e. sharp spin measurements in perpendicular directions. Proposition \ref{prop:wayofpovms2} implies that $ || \vec m \times \vec n || \leq 4 \left|\left| \big [U_\alpha, \So_{\vec{n}}\otimes \id \big] \right|\right|$. In other words, any unit vector $\vec{n}\in\R^3$ satisfying $4 \left|\left| \big [U_\alpha, \So_{\vec{n}}\otimes \id \big] \right|\right|<1$ implies a non-trivial constraint for the set of the realisable observables $\So_{\vec{m}}$.
 
For $\alpha=1$ the unitary $U_\alpha$ commutes with any quantity of the form
$ \left(\begin{array}{cc}
a & 0 \\
0 & b
\end{array}\right)\otimes \id,$ $a,b\in\R$, that is the observables measurable are exactly those compatible with $\So_{\hat z}$. On the other hand, for $\alpha <1$ the commutator $\big [U_\alpha, \So_{\vec{n}}\otimes \id \big]$ is non-vanishing for all unit vectors $\vec n\in\R^3$, but limitations may nevertheless be found for the realisability of the observables, as illustrated in Fig.\,\ref{fig:wayforpovms}. In the figure, the minimum value of $4 \left|\left| \big [U_\alpha, \So_{\vec{n}}\otimes \id \big] \right|\right|$ optimised over $\vec n$ is plotted against the parameter $\alpha$. In addition, the $xz$--cross-sections of those effects on the Bloch sphere satisfying $|| \vec m \times \vec n ||\leq 4 \left|\left| \big[U_\alpha, \So_{\vec{n}}\otimes \id \big] \right|\right|$ for the minimising $\vec n$, that is the effects that are in principle realisable with $U_\alpha$, are presented for five specific choices of $\alpha$: $0.8, 0.85, 0.9, 0.95$ and $1$. The full set of effects can be attained by rotating these cross-sections about the corresponding symmetry axes denoted in black.
\end{example}

\section{Application to quantum programming}

Looking at Eq.\,\eqref{eq:mntobs} one observes that altering the initial state of the probe may lead to measurements of different observables. We call this mapping from probe states to observables {\it quantum programming}, the fixed triplet $\langle \ki, \Zo, U\rangle$ a programmable quantum {\it multimeter} (also known as a programmable processor) and the variables of the multimeter {\it programming states}.

The obvious advantage of a quantum multimeter over a fixed set-up is that one does not have to build multiple measurement apparatuses in order to implement different observables. It is known, however, that no multimeter is universal, that is one cannot construct a multimeter $\langle \ki, \Zo, U\rangle$ that surjectively maps $\sk$ to the full set of observables on $\hi$. This follows from the fact that all unequal sharp observables demand mutually orthogonal programming states and, consequently, the number of programmable sharp observables is bounded by the dimension of the multimeter \cite{Dariano05, Tukiainen15, Tukiainen16}. As the proof of this result is quite concise, we will present it here for the readers' convenience.

\begin{proposition}\label{prop:ortho} Let $\langle \ki, \Zo, U\rangle$ be a multimeter realising sharp observables $\Ao_i:\Sigma\to\ph$ with programming states $\phi_i$, $i=1,2$, respectively. If $\Ao_1 \neq \Ao_2$ then $\ip {\phi_1} {\phi_2} =0$.
\end{proposition}
\begin{proof}
Let $X \in \Sigma$ be such that $\Ao_1(X) \neq \Ao_2(X)$. Since $ V_{\phi_2}^* V_{\phi_1} = \ip  {\phi_2} {\phi_1} $, we have
\begin{eqnarray}
\ip  {\phi_2} {\phi_1} \Ao_1(X) &=& V_{\phi_2}^* V_{\phi_1} V_{\phi_1}^* U^*\big( \id \otimes \Zo(X)\big) U V_{\phi_1} \nonumber \\
&=& V_{\phi_2}^*  U^*\big( \id \otimes \Zo(X)\big) U V_{\phi_2} V_{\phi_2}^* V_{\phi_1} \nonumber \\
&=& \ip {\phi_2} {\phi_1} \Ao_2(X) \, .
\end{eqnarray}
\end{proof}

\begin{example} The multimeter $\langle \C^2, \So_{\hat z}, U \rangle$ in Ex.\,\ref{ex:swap} can be programmed to realise the sharp observables $\So_{\hat z}$ and $\So_{\hat x}$ with programming states $\ket 0$ and $\ket 1$, respectively. Clearly, $\ip{0}{1}=0$.
\end{example}

It has been anticipated in Ref.\,\cite{Kimura08} that the WAY-theorem will set further restrictions for the programmable quantum multimeters. It is the purpose of this section to validate this expectation. To this end, consider a programmable multimeter $\langle \ki, \Zo, U \rangle$ that realises a sharp observable $\Ao_1:\Sigma \to \ph$ with a programming state $\phi_1\in\ki$, $||\phi_1||=1$. Proposition \ref{prop:wayofpovms1} then simplifies to 
\begin{eqnarray}\label{eq:wayprog}
\left|\left| \big[ \Ao_1(X), V_{\phi_1}^* L V_{\phi_1} \big] \right|\right| \leq  \left|\left| \big[ U^* \big( \id \otimes \Zo(X) \big) U, L \big] \right|\right| \, . \quad
\end{eqnarray}
Let $\Eo_2:\Sigma \to \lh$ be any other observable realisable with the multimeter $\langle \ki, \Zo, U\rangle$ and programming state $\phi_2\in\ki$, $||\phi_2||=1$, and fix an unitary operator $G$ on $\ki$ such that $G \phi_1 = \phi_2$. Since $V_{\phi_2} = \id \otimes G \, V_{\phi_1}$, we have $\Eo_2(Y) = V_{\phi_1}^* L(Y) V_{\phi_1} $ for a family of selfadjoint operators
\begin{eqnarray}
L(Y) := U_G^* \big( \id \otimes \Zo(Y) \big) U_G , \quad Y\in \Sigma \, ,
\end{eqnarray}
where $U_G = U \, \id \otimes G$.
Inserting $L(Y)$ in Eq.\,\eqref{eq:wayprog} results to the following proposition.

\begin{proposition}\label{prop:wayprog}
Let $\langle \ki, \Zo, U \rangle$ be a multimeter realising a sharp observable $\Ao_1:\Sigma\to\ph$ and an observable $\Eo_2:\Sigma \to \lh$ with programming states $\phi_1$ and $\phi_2$, respectively. For any unitary operator $G$ on $\ki$ satisfying $G \phi_1 = \phi_2$, the relation
\begin{eqnarray}\label{eq:wayprog2}
&&\left|\left| \big[ \Ao_1(X), \Eo_2(Y) \big] \right|\right| \nonumber \\
&\leq&  \left|\left| \big[ U^* \big( \id \otimes \Zo(X) \big) U, U_G^* \big( \id \otimes \Zo(Y) \big) U_G \big] \right|\right| 
\end{eqnarray}
holds for all $X,Y \in \Sigma$. 
\end{proposition}

Proposition \ref{prop:wayprog} confirms the existence of WAY-type limitations in quantum programming of observables. It is noteworthy, that the two multimeters $\langle \ki, \Zo, U\rangle$ and $\langle \ki, \Zo, U_G\rangle$ differ only by a local unitary transformation. Accordingly, they are {\it equivalent} in sense they both program exactly the same set of observables, only with different programming states \cite{Hillery10}. In this formalism Prop.\,\ref{prop:wayprog} relates the amount of (in)compatibility of the evolved pointers of the two multimeters with the (in)compatibility of the programmed observables. Such a relation can be useful in designing optimal multimeters, e.g. for purposes of measurement based quantum computing.

\section{Summary and discussion}
In summary, an approach to the theorem of Wigner, Araki and Yanase (WAY) was introduced which expresses the measurability limitations in the language of quantum incompatibility. Importantly, this formalism reveals a more intuitive and far more generally valid mathematical structure behind the WAY-theorem. In addition, two quantitative generalisations of WAY-type measurability restrictions to POVMs were presented. Finally, we demonstrated the potential of our results in applications of quantum programming.

Even though this analysis was focusing on the WAY-limitations of quantum observables, we wish to point out that the formalism can be straightforwardly extended also for general quantum devices, e.g. quantum channels or instruments. For example, the similarity between unitary channels and sharp observables noted in Ref.\,\cite{Tukiainen15} would allow one to find similar limitations to those proved here for quantum (unitary) channels. Although the details are beyond the scope of this paper and will be left as a topic for separate investigation, this approach could potentially lead to WAY-type limitations on quantum logic gates and computation that are more straightforward and general than those reported in Refs.\,\cite{Ozawa022, Karasawa07, Karasawa09}.

\section*{Acknowledgements}
The author would like to thank Prof. Paul Busch, Dr. Leon Loveridge, and his supervisors Dr. Teiko Heinosaari and Dr. Juha-Pekka Pellonp\"a\"a for useful comments and discussion and acknowledge financial support from the University of Turku Graduate School (UTUGS).


\end{document}